\documentclass[a4paper,UKenglish,cleveref,autoref]{lipics-v2019}

\title{Fast algorithms for handling diagonal constraints in timed automata}

\author{Paul Gastin}{LSV, CNRS, ENS Paris-Saclay, Universit\'e Paris--Saclay,
France}{paul.gastin@lsv.fr}{}{Supported by ANR project TickTac (ANR-18-CE40-0015)}

\author{Sayan Mukherjee}{Chennai Mathematical Institute, India}{sayanm@cmi.ac.in}{}{Infosys foundation (India) and Tata Consultancy Services - Innovation Labs (Pune, India)}

\author{B Srivathsan}{Chennai Mathematical Institute, India}{sri@cmi.ac.in}{}{Infosys foundation (India) and Tata Consultancy Services - Innovation Labs (Pune, India)}

\authorrunning{P. Gastin, S. Mukherjee, and B. Srivathsan}

\Copyright{Paul Gastin, Sayan Mukherjee, B Srivathsan}

\ccsdesc{Theory of Computation~Models of computation}
\ccsdesc{Models of computation~Timed and hybrid systems}

\keywords{Timed Automata, Reachability, Zones, Simulations,
  Diagonal constraints}

\funding{Supported by UMI ReLaX.}

\category{}

\relatedversion{}

\supplement{}

\nolinenumbers

\hideLIPIcs

\usepackage{enumitem}
\usepackage{amssymb}
\usepackage{amsmath}
\usepackage{amsfonts}
\usepackage{multirow}
\usepackage{wrapfig}
\usepackage{array}

\usepackage[textsize=small]{todonotes}

\newlength \reqlen
\usepackage{tikz}
\usetikzlibrary{calc,decorations.pathmorphing}

\usepackage[linesnumbered,vlined, boxed]{algorithm2e}

\usepackage{listings, multicol,mathtools,tocloft}

\lstdefinelanguage{algo}{%
  morekeywords={function,algorithm,push,pop,top,for,all,and,or,if,then,else,repeat,until,while,do,report,return,such,that,each,add,call,exit,let,Input,
  Output, procedure}
}

\newcommand{\Aa}{\mathcal{A}}
\newcommand{\Oo}{\mathcal{O}}

\newcommand{\Gg}{\mathcal{G}}

\renewcommand{\d}{\delta}

\newcommand{\true}{\operatorname{true}}
\newcommand{\false}{\operatorname{false}}
\newcommand{\Land}{\bigwedge}

\newcommand{\incl}{\subseteq}

\newcommand{\Rpos}{\mathbb{R}_{\ge 0}}
\newcommand{\Nat}{\mathbb{N}}
\newcommand{\Int}{\mathbb{Z}}

\newcommand{\PSPACE}{\operatorname{PSPACE}}

\newcommand{\xra}{\xrightarrow}

\newcommand{\vali}{\mathbf{0}}
\newcommand{\ZG}{ZG}
\newcommand{\tto}{\Rightarrow}

\newcommand{\fleq}{\preccurlyeq}

\newcommand{\lleq}{\mathrel{\triangleleft}}
\newcommand{\ggeq}{\mathrel{\triangleright}}

\newcommand{\regeq}{\simeq}

\renewcommand{\wp}{\operatorname{wp}}

\newcommand{\lu}{\fleq_{\scriptscriptstyle LU}}
\newcommand{\lud}{\fleq_{\scriptscriptstyle LU}^{\scriptscriptstyle
d}}

\newcommand{\lugg}{\fleq_{\scriptscriptstyle LU(\Gg)}}

\newcommand{\lugpar}[1]{\sqsubseteq_{\scriptscriptstyle {#1}}}
\newcommand{\lugdf}{\sqsubseteq_{\scriptscriptstyle \Gg^{-}}}

\newcommand{\lug}{\sqsubseteq_{\scriptscriptstyle \Gg}}
\newcommand{\lup}{\sqsubseteq_{\scriptscriptstyle \varphi}}
\newcommand{\lugo}{\sqsubseteq_{\scriptscriptstyle \Gg_1}}
\newcommand{\lugl}{\sqsubseteq_{\scriptscriptstyle \Gg}^*}
\newcommand{\luglu}{\sqsubseteq_{\scriptscriptstyle
\Gg}^{\scriptscriptstyle LU}}
\newcommand{\lugqlu}{\sqsubseteq_{\scriptscriptstyle
\Gg(q)}^{\scriptscriptstyle LU}}

\newcommand{\lugpl}{\sqsubseteq_{\scriptscriptstyle \Gg'}^*}

\newcommand{\lugone}{\sqsubseteq_{\scriptscriptstyle \Gg(q_1)}}
\newcommand{\gq}{\sqsubseteq_{\scriptscriptstyle \Gg(q)}}
\newcommand{\gqone}{\sqsubseteq_{\scriptscriptstyle \Gg(q_1)}}
\newcommand{\as}{\fleq_{\Aa}}
\newcommand{\aslu}{\fleq^{\scriptscriptstyle LU}_{\scriptscriptstyle \Aa}}

\newcommand{\lugp}{\sqsubseteq_{\scriptscriptstyle \Gg'}}

\newcommand{\down}[1]{{\downarrow}{#1}}
\newcommand{\downlu}[1]{{\downarrow_{\scriptscriptstyle LU}}{#1}}

\newcommand{\Extra}{\mathsf{Extra}}

\newcommand{\ExtraLUp}{\operatorname{Extra}_{\scriptscriptstyle LU}^+}
\newcommand{\ExtraM}{\operatorname{Extra}_{\scriptscriptstyle M}}

\DeclareSymbolFont{symbolsC}{U}{txsyc}{m}{n}
\SetSymbolFont{symbolsC}{bold}{U}{txsyc}{bx}{n}
\DeclareFontSubstitution{U}{txsyc}{m}{n}
\DeclareMathSymbol{\multimapboth}{\mathrel}{symbolsC}{19}

\makeatletter
\newcommand{\xdashrightarrow}[2][]{\ext@arrow 0359\rightarrowfill@@{#1}{#2}}
\newcommand{\xdashleftarrow}[2][]{\ext@arrow 3095\leftarrowfill@@{#1}{#2}}
\newcommand{\xdashleftrightarrow}[2][]{\ext@arrow 3359\leftrightarrowfill@@{#1}{#2}}
\def\rightarrowfill@@{\arrowfill@@\relax\relbar\rightarrow}
\def\leftarrowfill@@{\arrowfill@@\leftarrow\relbar\relax}
\def\leftrightarrowfill@@{\arrowfill@@\leftarrow\relbar\rightarrow}
\def\arrowfill@@#1#2#3#4{%
  $\m@th\thickmuskip0mu\medmuskip\thickmuskip\thinmuskip\thickmuskip
   \relax#4#1
   \xleaders\hbox{$#4#2$}\hfill
   #3$%
}
\makeatother

\usepackage{rotating}
\newcounter{sarrow}

\begin{document}

\maketitle

\begin{abstract}
  A popular method for solving reachability in timed automata proceeds
  by enumerating reachable sets of valuations represented as zones. A
  na\"{i}ve enumeration of zones does not terminate. Various termination
  mechanisms have been studied over the years. Coming up with
  efficient termination mechanisms has been remarkably more
  challenging when the automaton has diagonal constraints in
  guards.

 In this paper, we propose a new termination mechanism for timed
 automata with diagonal constraints based on a new
 simulation relation between zones.  Experiments with an
 implementation of this simulation show
 significant gains over existing methods.
\end{abstract}

\section{Introduction}
\label{sec:introduction}

Timed automata have emerged as a popular model for systems with
real-time constraints~\cite{Alur:TCS:1994}.
Timed automata are finite automata extended with real-valued variables
called \emph{clocks}. All clocks are assumed to start at $0$, and
increase at the same rate. Transitions of the automaton can make use
of these clocks to disallow behaviours which violate timing
constraints. This is achieved by making use of \emph{guards} which are
constraints of the form $x \le 5$, $ x - y \ge 3$, $ y > 7$, etc where
$x, y$ are
clocks.
A transition guarded by $x \le 5$ says that it can be fired only when
the value of clock $x$ is $\le 5$. Another important feature is the
\emph{reset} of clocks in transitions. Each transition can specify a
subset of clocks whose values become $0$ once the transition is
fired. The combination of guards and resets allows to track timing
distance between events. A basic question that forms the core of timed
automata technology is \emph{reachability}: given a timed automaton,
does there exist an execution from its initial state to a final
state. This question is known to be decidable~\cite{Alur:TCS:1994}.
Various algorithms for this problem have been studied over the years
and have been implemented in
tools~\cite{Larsen:1997:UPPAAL,Bengtsson:Springer:2004,LTSMin,KRONOS,Wang:2004:RED}.

Since the clocks are real valued variables, the space of
configurations of a timed automaton (consisting of a state and a
valuation of the clocks) is infinite and an explicit enumeration is
not possible. The earliest solution to reachability was to partition
this space into a finite number of \emph{regions} and build a region
graph that provides a finite abstraction of the behaviour of the timed
automaton~\cite{Alur:TCS:1994}. However, this solution was not
practical. Subsequent works introduced the use of
\emph{zones}~\cite{Daws:TACAS:1998}. Zones are special sets of clock
valuations with efficient data structures and manipulation
algorithms~\cite{Bengtsson:Springer:2004}. Within zone based
algorithms, there is a division: forward analysis versus backward
analysis. The current industry strength tool
UPPAAL~\cite{Larsen:1997:UPPAAL} implements a forward analysis
approach, as this works better in the presence of other discrete data
structures used in UPPAAL
models~\cite{Bouyer:2004:forwardanalysis}. We focus on this forward
analysis approach using zones in this paper.

The forward analysis of a timed automaton essentially enumerates sets
of reachable configurations stored as zones. Some extra care needs to
be taken for this enumeration to terminate. Traditional development of
timed automata made use of \emph{extrapolation} operators over zones
to ensure termination. These are functions which map a zone to a
bigger zone. Importantly, the range of these functions is finite.
The goal was to come up with extrapolation operators which are sound:
adding these extra valuations should not lead to new behaviours. This
is where the role of \emph{simulations} between configurations was
studied and extrapolation operators based on such simulations were
devised~\cite{Daws:TACAS:1998}.  A certain extrapolation operation,
which is now known as $\Extra_M$~\cite{Behrmann:STTT:2006} was
proposed and reachability using $\Extra_M$ was implemented in
tools~\cite{Daws:TACAS:1998}.

A seminal paper by Bouyer~\cite{Bouyer:2004:forwardanalysis} revealed
that $\Extra_M$ is not correct in the presence of \emph{diagonal
  constraints} in guards. These are constraints of the form
$x - y \lleq c$ where $\lleq$ is either $<$ or $\le$, and $c$ is an
integer. Moreover, it was proved that no such extrapolation operation
would be correct when there are diagonal constraints
present.
It was shown that for automata without diagonal constraints
(henceforth referred to as diagonal-free automata), the extrapolation
works. After this result, developments in timed automata reachability
focussed on the class of diagonal-free automata
\cite{Behrmann:STTT:2006,Behrmann:TACAS:2003,Herbreteau:IandC:2016,Herbreteau:2015:FORMATS},
and diagonal constraints were mostly sidelined.  All these
developments have led to quite efficient algorithms for diagonal-free
timed automata.

Diagonal constraints are a useful modeling feature and occur naturally
in certain problems, especially
scheduling~\cite{DBLP:conf/tacas/FersmanPY02} and logic-automata
translations~\cite{Ferrere:2018:FM,Hsi-Ming:2018:Logics}. It is
however known that they do not add any expressive power: every timed
automaton can be converted into a diagonal-free timed
automaton~\cite{Berard:1998:FundInf}. This conversion suffers from an
exponential blowup, which was later shown to be unavoidable: diagonal
constraints could potentially give exponentially more succinct
models~\cite{Bouyer:2005:Conciseness}. Therefore, a good forward
analysis algorithm that works directly on a timed automaton with
diagonal constraints would be handy. This is the subject of this
paper.

\medskip\noindent \emph{Related work.} The first attempt at such an
algorithm was to split the (extrapolated) zones with respect to the
diagonal constraints present in the
automaton~\cite{Bengtsson:Springer:2004}. This gave a correct
procedure, but since zones are split, an enumeration starts from each
small zone leading to an exponential blow-up in the number of visited
zones. A second attempt was to do a more refined conversion into a
diagonal free automaton by detecting ``relevant''
diagonals~\cite{Bouyer:2005:diagonal-refinement,Reynier:toolreport} in
an iterative manner. In order to do this, special data structures
storing sets of sets of diagonal constraints were utilized. A more
recent attempt~\cite{Gastin:2018:CONCUR} extends the works
\cite{Behrmann:STTT:2006} and \cite{Herbreteau:IandC:2016} on
diagonal-free automata to the case of diagonal constraints. All the
approaches suffer from either a space or time bottleneck and are
incomparable to the efficiency and scalability of tools for
diagonal-free automata.

\medskip\noindent \emph{Our contributions.} The goal of this paper is
to come up with fast algorithms for handling diagonal
constraints. Since the extrapolation based approach is a dead end, we
work with simulation between zones directly, as in
\cite{Herbreteau:IandC:2016} and \cite{Gastin:2018:CONCUR}. We propose
a new simulation relation between zones that is correct in the
presence of diagonal constraints
(Section~\ref{sec:new-simul-relat}). We give an algorithm to test this
simulation between zones (Section~\ref{sec:algorithm-z-lug}). We have
incorporated this simulation test in the tool TChecker \cite{tchecker} that checks
reachability for timed automata, and compared our results with the
state-of-the-art tool UPPAAL. Experiments show an encouraging gain,
both in the number of zones enumerated and in the time taken by the
algorithm, sometimes upto four orders of magnitude
(Section~\ref{sec:experiments}). The main advantage of our approach is
that it does not split zones, and furthermore it leverages the
optimizations studied for diagonal-free
automata.

From a technical point of view, our presentation does not make use of
regions and instead works with valuations, zones and simulation
relations. We think that this presentation provides a clearer
perspective - as a justification of this claim, we extend our
simulation to timed automata with general updates of the form $x: = c$
and $x := y + d$ in transitions (where $x,y$ are clocks and $c, d$ are
constants) in a rather natural manner (Section~\ref{sec:updates}). In general, reachability for
timed automata with updates is undecidable~\cite{Bouyer:2004:Updateable}. Some
decidable cases have been proposed for which the algorithms are based
on regions. For decidable subclasses containing diagonal constraints,
no zone based approach has been studied. Our proposed method includes
these classes, and also benefits from zones and standard
optimizations studied for diagonal-free automata.

\newcommand{\Ss}{\mathcal{S}}
\newcommand{\elapse}[1]{\overrightarrow{#1}}
\section{Preliminaries}
\label{sec:preliminaries}

Let $\Nat$ be the set of natural numbers, $\Rpos$ the set of
non-negative reals and $\Int$ the set of integers. Let $X$ be a finite
set of variables ranging over $\Rpos$, called \emph{clocks}. Let
$\Phi(X)$ denote the set of constraints $\varphi$ formed using the
following grammar:
$\varphi := x \lleq c~ \mid ~ c \lleq x ~\mid~ x - y \lleq d ~\mid~
\varphi \land \varphi$, where $x, y \in X$, $c \in \Nat$, $d \in \Int$
and ${\lleq} \in \{<, \le\}$.  Constraints of the form $x \lleq c$ and
$c \lleq x$ are called \emph{non-diagonal constraints} and those of
the form $x - y \lleq c$ are called \emph{diagonal constraints}. We
have adopted a convention that in non-diagonal constraints $x \lleq c$
and $c \lleq x$, the constant $c$ is restricted to $\Nat$.  A
\emph{clock valuation} $v$ is a function which maps every clock
$x \in X$ to a real number $v(x) \in \Rpos$. A valuation is said to
satisfy a guard $g$, written as $v \models g$ if replacing every $x$
in $g$ with $v(x)$ makes the constraint $g$ true. For $\d \in \Rpos$
we write $v + \d$ for the valuation which maps every $x$ to
$v(x) + \d$. Given a subset of clocks $R \incl X$, we write $[R]v$ for
the valuation which maps each $x \in R$ to $0$ and each $x \not \in R$
to $v(x)$.

A \emph{timed automaton} $\Aa$ is a tuple $(Q, X, q_0, T, F)$ where
$Q$ is a finite set of states, $X$ is a finite set of clocks,
$q_0 \in Q$ is the initial state, $F \incl Q$ is a set of accepting
states and $T \in Q \times \Phi(X) \times 2^X \times Q$ is a set of
transitions. Each transition $t \in T$ is of the form $(q, g, R, q')$
where $q$ and $q'$ are respectively the source and target states, $g$
is a constraint called the \emph{guard}, and $R$ is a set of clocks
which are \emph{reset} in $t$. We call a timed automaton
\emph{diagonal-free} if guards in transitions do not use diagonal
constraints.

A \emph{configuration} of $\Aa$ is a pair $(q, v)$ where $q \in Q$ and
$v$ is a valuation. The semantics of a timed automaton is given by a
transition system $\Ss_\Aa$ whose states are the configurations of
$\Aa$. Transitions in $\Ss_\Aa$ are of two kinds: \emph{delay}
transitions are given by $(q, v) \xra{\d} (q, v + \d)$ for all
$\d \ge 0$, and \emph{action} transitions are given by
$(q, v) \xra{t} (q',v')$ for each $t := (q, g, R, q')$, if
$v \models g$ and $v' = [R]v$. We write $\xra{\d, t}$ for a sequence
of delay $\d$ followed by action $t$. A run of $\Aa$ is an alternating
sequence of delay-action transitions starting from the initial state
$q_0$ and the initial valuation $\vali$ which maps every clock to $0$:
$(q_0, \vali) \xra{\d_0, t_0} (q_1, v_1) \xra{\d_1,t_1} \cdots (q_n,
v_n)$. A run of the above form is said to be accepting if the last
state $q_n \in F$.  The \emph{reachability problem} for timed automata
is the following: given an automaton $\Aa$, decide if there exists an
accepting run. This problem is known to be
$\PSPACE$-complete~\cite{Alur:TCS:1994}. Since the semantics $\Ss_\Aa$
is infinite, solutions to the reachability problem work with a finite
abstraction of $\Ss_\Aa$ that is sound and complete. Before we explain
one of the popular solutions to reachability, we state a result which
allows to convert every timed automaton into a diagonal-free timed
automaton.

\begin{theorem}~\cite{Berard:1998:FundInf} \label{thm:diagonal-to-diagonalfree}
  For every timed automaton $\Aa$, there exists a diagonal-free timed
  automaton $\Aa_{df}$ s.t. there is a bijection between runs of $\Aa$
  and $\Aa_{df}$. The number of states in $\Aa_{df}$ is $2^d\cdot n$
  where $d$ is the number of diagonal constraints and $n$ is the
  number of states of $\Aa$.
\end{theorem}

\begin{proof}(Sketch) We give the construction to remove one diagonal.
  This can be applied iteratively to remove all diagonal
  constraints. Let $x - y \lleq c$ be a diagonal constraint appearing
  in $\Aa$. Initially, the value of $x - y$ is $0$ and hence
  $x - y \lleq c$ is true iff $0 \lleq c$. The value $x - y$ does not
  change by time elapse. It can change only due to resets. Therefore,
  the idea is to use the states of the automaton to maintain if the
  constraint is true or false. Let $\Aa'$ be the new automaton which
  does not contain $x - y \lleq c$.

  The states of $\Aa'$ are of the form $Q \times \{0, 1\}$. For every
  transition $(q, g, R, q')$ of $\Aa$ we have the following
  transitions in $\Aa'$, where $g' = g \setminus \{x -y \lleq c\}$,
  and $b = 1$ if $g$ contains $x - y \lleq c$, and $b \in \{0,1\}$
  otherwise:
  \begin{itemize}
  \item if $\{x, y\} \cap R = \emptyset$, then
    $(q, b) \xra[R]{~g'~} (q', b)$
  \item if $\{x, y\} \subseteq R$, then
    $(q, b) \xra[R]{~g'~} (q', b')$ where $b' = 1$ if $0 \lleq c$, and
    $b'=0$ otherwise
  \item if $x \in R$ and $y \notin R$, there are transitions
    $(q, b) \xra[R]{~g' \land -y \lleq c~} (q', 1)$ and
    $(q, b) \xra[R]{~g' \land -y \ggeq_1 c~} (q', 0)$ where $\ggeq_1$
    is $>$ if $\lleq$ equals $\le$; and $\ge$ otherwise. The guards
    $-y \lleq c$ and $-y \ggeq_1 c$ can be suitably removed if they
    are either trivially true or false, and if not, rewritten to the
    form $y \sim d$ where $d \ge 0$.
  \item if $x \not \in R$ and $y \in R$, the construction is similar
    to the above case.
  \end{itemize}

  The initial state is $(q_0, 0 \lleq c)$ and the accepting states are
  $F \times \{0, 1\}$. Note that there is a one-one correspondence
  between runs of $\Aa$ and $\Aa'$ and moreover $\Aa'$ does not have
  the diagonal constraint $x -y \lleq c$.
\end{proof}

The above theorem allows to solve the reachability of a timed
automaton $\Aa$ by first converting it into the diagonal free
automaton $\Aa_{df}$ and then checking reachability on
$\Aa_{df}$. However, this conversion comes with a systematic
exponential blowup (in terms of the number of diagonal constraints
present in $\Aa$). It was shown in \cite{Bouyer:2005:Conciseness} that
such a blowup is unavoidable in general. We will now recall the
general algorithm for analyzing timed automata, and then move into
specific details which depend on whether the automaton has diagonal
constraints or not.

\subsubsection*{Zones and simulations}
Fix a timed automaton $\Aa$ with clock set $X$ for the rest of the
discussion in this section. As the space of valuations of $\Aa$ is
infinite, algorithms work with sets of valuations called
\emph{zones}. A zone is set of clock valuations given by a conjunction
of constraints of the form $x - y \lleq c$, $x \lleq c$ and
$c \lleq x$ where $c \in \Int$ and ${\lleq} \in \{<, \le\}$, for example
the solutions of $x - y < 5 \land y \le 10$ is a zone. The transition
relation over configurations $(q,v)$ is extended to $(q, Z)$ where $Z$
is a zone. We define the following operations on zones given a guard
$g$ and a set of clocks $R$: time elapse
$\elapse{Z} = \{ v + \d~|~ v \in Z, \d \ge 0\}$; guard intersection
$Z \land g := \{v~|~v \in Z \text{ and } v \models g \}$ and reset
$[R]Z := \{ [R]v~|~ v \in Z\}$. It can be shown that all these
operations result in zones. Zones can be efficiently represented and
manipulated using Difference Bound Matrices (DBMs)
\cite{Dill:1990:DBM}.

The \emph{zone graph} $\ZG(\Aa)$ of timed automaton $\Aa$ is a
transition system whose nodes are of the form $(q, Z)$ where $q$ is a
state of $\Aa$ and $Z$ is a zone. For each transition
$t:= (q, g, R, q')$ of $\Aa$, and each zone $(q, Z)$ there is a
transition $(q, Z) \tto^t (q', Z')$ where
$Z' = \elapse{[R](Z \land g)}$. The initial node is $(q_0, Z_0)$ where
$q_0$ is the initial state of $\Aa$ and
$Z_0 = \{ \vali + \d~|~\d \ge 0\}$ is the zone obtained by elapsing an
arbitrary delay from the initial valuation. A path in the zone graph
is a sequence
$(q_0, Z_0) \tto^{t_0} (q_1, Z_1) \tto^{t_1} \cdots \tto^{t_{n-1}}
(q_n, Z_n)$ starting from the initial node. The path is said to be
accepting if $q_n$ is an accepting state. The zone graph is known to
be sound and complete for reachability.

\begin{theorem}\cite{Daws:TACAS:1998}
  $\Aa$ has an accepting run iff $\ZG(\Aa)$ has an accepting path.
\end{theorem}

This does not yet give an algorithm as the zone graph $\ZG(\Aa)$ is
still not finite. Moreover, there are examples of automata for which
the reachable part of $\ZG(\Aa)$ is also infinite: starting from the
initial node, applying the successor computation leads to infinitely
many zones. Two different approaches have been studied to get
finiteness, both of them based on the usage of \emph{simulation
  relations}. A (time-abstract) simulation relation $\fleq$ between
configurations of $\Aa$ is a reflexive and transitive relation such
that $(q, v) \fleq (q', v')$ implies $q = q'$ and (1) for every
$\d \ge 0$, there exists $\d' \ge 0$ such that
$(q, v+\d) \fleq (q, v' + \d')$ and (2) for every transition $t$ of
$\Aa$, if $(q, v) \xra{t} (q_1, v_1)$ then
$(q, v') \xra{t} (q_1, v_1')$ such that
$(q_1, v_1) \fleq (q_1,v_1')$\label{ta-simulation}. We say
$v \fleq v'$, read as $v$ is simulated by $v'$ if
$(q,v) \fleq (q, v')$ for all states $q$. The simulation relation can
be extended to zones: $Z \fleq Z'$ if for every $v \in Z$ there exists
$v' \in Z'$ such that $v \fleq v'$. We write $\down{Z}$ for
$\{ v ~|~ \exists v' \in Z \text{ s.t. } v \fleq v'\}$. The simulation
relation $\fleq$ is said to be finite if the function mapping zones
$Z$ to the down sets $\down{Z}$ has finite range. We now recall a
specific simulation relation which was shown to be finite and correct
for diagonal-free timed automata. Current algorithms and tools for
diagonal-free automata are based on this simulation. The conditions
required for $v \lu v'$ ensure that when all lower bound constraints
$c \lleq x$ satisfy $c \le L(x)$ and all upper bound constraints $x
\lleq c$ satisfy $c \le U(x)$, whenever $v$ satisfies a constraint,
$v'$ will also satisfy it.

\begin{definition}[LU-bounds and the relation
  $\lu$~\cite{Behrmann:STTT:2006,Herbreteau:IandC:2016}]
  \label{def:lu-sim}
  An $LU$-bounds function is a pair of functions
  $L: X \mapsto \Nat \cup \{ -\infty\}$ and
  $U: X \mapsto \Nat \cup \{-\infty\}$ that map each clock to either a
  non-negative constant or $-\infty$. Given an $LU$-bounds function,
  we define $v \lu v'$ for valuations $v, v'$ if for every clock
  $x \in X$:
  \begin{align*}
    \text{ $v'(x) < v(x)$ implies $L(x) < v'(x)$ } \quad \text{ and }
    \quad \text{ $v(x) < v'(x)$ implies $U(x) < v(x)$. }
  \end{align*}
\end{definition}

The $\lu$ simulation is extended to zones as mentioned above. Let us
write $\downlu{Z}$ for
$\{ v~|~ \exists v' \in Z \text{ with } v \lu v'\}$.

\begin{theorem}\cite{Behrmann:STTT:2006}\label{thm:lu-is-correct-diagonal-free}
  Let $\Aa$ be a diagonal-free timed automaton. Let $LU$ be a bounds
  function such that $c \le L(x)$ for every (lower-bound) guard
  $c \lleq x$ appearing in $\Aa$ and $c \le U(x)$ for every
  (upper-bound) guard $x \lleq c$ of $\Aa$. Then, the relation $\lu$
  is a finite simulation on the configurations of
  $\Aa$.
\end{theorem}

\subsubsection*{Reachability in diagonal-free timed automata}
\label{sec:reach-diag-free}

A natural method to get finiteness of the zone graph is to prune the
zone graph computation through simulations $Z \lu Z'$: do not explore
a node $(q,Z)$ if there is an already visited node $(q,Z')$ such that
$Z \lu Z'$. Since these simulation tests need to be done often during
the zone graph computation, an efficient algorithm for performing this
test is crucial. Note that $Z \lu Z'$ iff $Z \incl
\downlu{Z'}$. However, it is known that the set $\downlu{Z'}$ is not
necessarily a zone~\cite{Behrmann:STTT:2006}, and hence no simple zone
inclusions are applicable.  The first algorithms for timed automata
followed a different approach, which we call the \emph{extrapolation}
approach. A function $\ExtraLUp$~\cite{Behrmann:STTT:2006} on zones
was defined, which had two nice properties - (1)
$\ExtraLUp(Z) \incl \downlu{Z}$ and (2) $\ExtraLUp(Z)$ is a zone for
all $Z$. These properties give rise to an algorithm that performs only
efficient zone operations: successor computations and zone inclusions.

\label{algo:extrapolation-reachability}
\medskip\noindent \emph{Reachability algorithm using zone
  extrapolation.} The input to the algorithm is a timed automaton
$\Aa$.  The algorithm maintains two lists, Passed and
Waiting. Initially, the node $(q_0, \ExtraLUp(Z_0))$ is added to the
Waiting list (recall that $(q_0, Z_0)$ is the initial node of the zone
graph $ZG(\Aa)$).  Wlog.\ we assume that $q_0$ is not accepting.  The
algorithm repeatedly performs the following steps:
\begin{description}
\item[Step 1.] If Waiting is empty, then return ``$\Aa$ has no
  accepting run''; else pick (and remove) a node $(q, Z)$ from
  Waiting. Add $(q, Z)$ to Passed.
\item[Step 2.] For each transition $t:=(q, g, R, q_1)$, compute the
  successor $(q, Z) \tto^t (q_1, Z_1)$: if $Z_1 \neq \emptyset$
  perform the following operations - if $q_1$ is accepting, return
  ``$\Aa$ has an accepting run''; else compute
  $\hat{Z}_1 := \ExtraLUp(Z_1)$ and check if there exists a node
  $(q_1, Z_1')$ in Passed or Waiting such that $\hat{Z}_1 \incl Z_1'$:
  if yes, ignore the node $(q_1, \hat{Z}_1)$, otherwise add
  $(q_1, \hat{Z}_1)$ to Waiting.
\end{description}

The correctness of this algorithm follows from
Theorem~\ref{thm:lu-is-correct-diagonal-free} and the fact that
$\ExtraLUp(Z) \incl
\downlu{Z}$~\cite{Behrmann:STTT:2006}. Intuitively, the additional
valuations on top of $Z$ in $\ExtraLUp(Z)$ belong to $\downlu{Z}$, and
hence are simulated by some valuation in $Z$. This therefore entails
that there are no new states visited due to these extra valuations.
More recently, an $\Oo(|X|^2)$ algorithm for $Z \lu Z'$ was proposed
\cite{Herbreteau:IandC:2016}. This algorithm has the same complexity
as the test for inclusion between zones. Additionally, the $\lu$ test
was a simple extension of the zone inclusion algorithm, to account for
the $LU$ bounds. This gives an algorithm that can avoid explicit
extrapolations, and also use a better set $\downlu{Z'}$ to prune the
search.

\label{algo:simulation-reachability}
\medskip\noindent \emph{Reachability algorithm using simulations.}  The initial
node $(q_0, Z_0)$ is added to the Waiting list.  Wlog.\ we assume that
$q_0$ is not accepting.  The algorithm repeatedly performs the
following steps:
\begin{description}
\item[Step 1.] If Waiting is empty, then return ``$\Aa$ has no
  accepting run''; else pick (and remove) a node $(q, Z)$ from
  Waiting. Add $(q, Z)$ to Passed.
\item[Step 2.] For each transition $t:=(q, g, R, q_1)$, compute the
  successor $(q, Z) \tto^t (q_1, Z_1)$: if $Z_1 \neq \emptyset$
  perform the following operations - if $q_1$ is accepting, return
  ``$\Aa$ has an accepting run''; else check if there exists a node
  $(q_1, Z_1')$ in Passed or Waiting such that $Z_1 \lu Z_1'$: if yes,
  ignore the node $(q_1, Z_1)$, otherwise add $(q_1, Z_1)$ to Waiting.
\end{description}

\subsubsection*{Reachability in the presence of diagonal constraints}
\label{sec:reach-timed-autom-diagonals}

The $\lu$ relation is no longer a
simulation when diagonal constraints are present. Moreover, it was
shown in \cite{Bouyer:2004:forwardanalysis} that no extrapolation
operator (along the lines of $\ExtraLUp$) can work in the presence of
diagonal constraints. The first option to deal with diagonals is to
use Theorem~\ref{thm:diagonal-to-diagonalfree} to get a diagonal free
automaton and then apply the methods discussed previously. One problem
with this is the systematic exponential blowup introduced in the
number of states of the resulting automaton. Another problem is to get
diagnostic information: counterexamples need to be translated back to
the original automaton \cite{Bengtsson:Springer:2004}. Various methods
have been studied to circumvent the diagonal free conversion and
instead work on the automaton with diagonal constraints directly. We
recall the approach used in the state-of-the-art tool UPPAAL below.

\medskip\noindent \emph{Zone splitting
  \cite{Bengtsson:Springer:2004}.} The paper introducing timed
automata gave a notion of equivalence between valuations
$v \regeq_M v'$ parameterized by a function $M$ mapping each clock $x$
to the maximum constant $M$ among the guards of the automaton that
involve $x$. This equivalence is a finite simulation for diagonal-free
automata. Equivalence classes of $\regeq_M$ are called regions. This
was extended to the diagonal case by \cite{Bengtsson:Springer:2004}
as: $v \regeq_M^d v'$ if $v \regeq_Mv'$ and for all diagonal
constraints $g$ present in the automaton, if $v \models g$ then
$v' \models g$. The $\regeq_M^d$ relation splits the regions further,
such that each region is either entirely included inside $g$, or
entirely outside $g$ for each $g$. The next step is to use this notion
of equivalence in zones. The paper \cite{Bengtsson:Springer:2004}
follows the extrapolation approach: to each zone $Z$, an extrapolation
operation $\ExtraM(Z)$ is applied; this adds some valuations which are
$\regeq_M$ equivalent to valuations in $Z$; then it is further split
into multiple zones, so that each small zone is either inside $g$ or
outside $g$ for each diagonal constraint $g$. If $d$ is the number of
diagonal constraints present in the automaton, this splitting process
can give rise to $2^d$ zones for each zone $Z$. From each small zone,
the zone graph computation is started. Essentially, the exponential
blow-up at the state level which appeared in the diagonal-free
conversion now appears in the zone level.

\medskip\noindent \emph{Finding relevant diagonal constraints
  \cite{Bouyer:2005:diagonal-refinement}.} This approach wants to perform the
diagonal-free conversion, however restricted to a subset of diagonal
constraints. Initially, none of the diagonal constraints are removed
and the extrapolated zone graph (using $\ExtraM$) is computed. Since
this graph contains more valuations, this is an over-approximation: if
there is a path over these extrapolated zones to an accepting state,
it might not be an actual path over concrete zones. Hence, each time
an accepting state is reached, a concrete zone computation is done to
check if it is spurious. If it is indeed spurious, the path was
disabled in between by some guard. The goal is to now find all
diagonal constraints along the path that were responsible for this
disability: the authors propose to do this by maintaining for each
zone $Z$, and each pair of clocks $x, y$, a set of diagonal
constraints that were used to get the current value of $x - y$ in
$Z$. In order to get smaller sets of diagonal constraints, the authors
in fact maintain a set-of-sets so that any set in this set-of-sets can
be picked as a candidate. Once this candidate set $S$ is picked, the
diagonal free conversion restricted to $S$ is done. This is the
refinement. Again, an extrapolated zone graph is computed, and the
process is continued till there are no false positives. Since each
refinement removes a diagonal constraint, in the end the algorithm
will compute the full diagonal-free conversion. This method works well
when not many diagonal constraints result in false positives.

\medskip\noindent \emph{$LU$ extended to diagonals.} Recently,
\cite{Gastin:2018:CONCUR} gave an extension of $LU$ simulation, called
$\lud$, to the case of diagonals. This was parameterized by two
constants $LU$ associated to each $x - y$ for $x, y \in X$.  An
algorithm $Z \lud Z'$ was developed. However, the complexity of this
test is unsatisfactory:it was shown that deciding $Z \lud Z'$ is
NP-hard. Although this results in a tremendous decrease in the number
of nodes explored, the algorithm performs bad in terms of timing due
to the implementation of the simulation test.

In this paper, we propose a new simulation to handle diagonal
constraints. This has two advantages - using this avoids the blow-up
in the number of nodes arising due to zone splitting, and the
simulation test between zones has an efficient implementation and is
significantly quicker than the simulation of
\cite{Gastin:2018:CONCUR}.

\section{A new simulation relation}
\label{sec:new-simul-relat}

In this section, we start with a definition of a relation between
timed automata configurations, which in some sense ``declares''
upfront what we need out of a simulation relation that can be utilized
in a reachability algorithm. As we proceed, we will make its
description more concrete and give an effective simulation algorithm
between zones, that can be implemented. Fix a set of clocks $X$. This
generates constraints $\Phi(X)$.

\begin{definition}[the relation $\lug$]
  \label{def:lug-simulation}
  Let $\Gg$ be a (finite or infinite) set of
  constraints. We say $v \lug v'$ if for all $\varphi \in \Gg$ and
  all $\d \ge 0$,  $v + \d \models \varphi$ implies
  $v' + \d \models \varphi$.
\end{definition}

Our goal is to utilize the above relation in a simulation (as defined
in Page~\pageref{ta-simulation}) for a timed automaton. Directly from
the definition, we get the following lemma which shows that the $\lug$
relation is preserved under time elapse.

\begin{lemma}\label{lem:lug-over-time}
  If $v \lug v'$, then $v + \d \lug v' + \d$ for all $\d \ge 0$.
\end{lemma}

The other kind of transformation over valuations is resets. Given sets
of guards $\Gg_1$, $\Gg$ and a set of clocks $R$, we want to find
conditions on $\Gg_1$ and $\Gg$ so that if $v \lugo v'$ then
$[R]v \lug [R]v'$. To do this, we need to answer this question: what
guarantees should we ensure for $v, v'$ (via $\Gg_1$) so that
$[R]v \lug [R]v'$.  This motivates the next definition.

\begin{definition}[weakest pre-condition of $\lug$ over resets]\label{def:wp-resets}
  For a constraint $\varphi$ and a set of clocks $R$, we define a
  set of constraints $\wp(\lup, R)$ as follows: when $\varphi$ is of the form
  $x \lleq c$ or $c \lleq x$, then $\wp(\lup, R)$ is empty if
  $x \in R $ and is $\{\varphi\}$ otherwise; when $\varphi$ is a diagonal
  constraint $x - y \lleq c$, then $\wp(\lup, R)$ is:
  \begin{itemize}[nosep]
  \item $\{x - y \lleq c\}$ if $\{x, y\} \cap R = \emptyset$
  \item $\{x \lleq c\}$ if $y \in R$, $x \not \in R$ and $c \ge 0$
  \item $\{- c \lleq y\}$ if $x \in R$, $ y \not \in R$ and $-c \ge 0$
  \item empty, otherwise.
  \end{itemize}
  For a set of guards $\Gg$, we define
  $\wp(\lug, R) := \bigcup_{\varphi \in \Gg} \wp(\lup, R)$.
\end{definition}

Note that the relation $\lug$ is parameterized by a set of constraints.
Additionally, we desire this set to be finite, so that the relation can be used
in an algorithm.  We need to first link an automaton $\Aa$ with such a set of
constraints.  One way to do it is to take the set of all guards present in the
automaton and to close it under weakest pre-conditions with respect to all
possible subsets of clocks.  A better approach is to consider a set of
constraints for each state, as in \cite{Behrmann:TACAS:2003} where the
parameters for extrapolation (the maximum constants appearing in guards) are
calculated at each state.

\begin{definition}[State based guards]
  \label{def:local-guards}
  Let $\Aa = (Q, X, q_0, T, F)$ be a timed automaton. We associate a
  set of guards $\Gg(q)$ for each state $q \in Q$, which is the least
  set of guards (for the coordinate-wise subset inclusion order) such
  that for every transition $(q, g, R, q_1)$: the guard $g$ and the
  set $\wp(\lugone, R)$ are present in $\Gg(q)$. More precisely,
  $\{\Gg(q)\}_{q \in Q}$ is the least solution to the following set of
  equations written for each $q \in Q$:
  \begin{align*}
    \Gg(q) = \bigcup_{(q, g, R, q_1) \in T} \{g\} \cup \wp(\lugone, R)
  \end{align*}
\end{definition}

All constraints present in the set $\wp(\lugone, R)$ contain constants
which are already present in $\lugone$. The least solution to the
above set of equations can therefore be obtained by a fixed point
computation which starts with $\Gg(q)$ set to
$\bigcup_{(q, g, R, q_1) \in T} \{g\}$ and then repeatedly updates the
weakest-preconditions. Since no new constants are generated in this
process, the fixed point computation terminates. We now have the
ingredients to define a simulation relation over configurations of a
timed automaton with diagonal constraints.

\begin{definition}[$\Aa$-simulation]
  \label{def:a-simulation}
  Let $\Aa = (Q, X, q_0, T, F)$ be a timed automaton and let the set
  of guards $\Gg(q)$ of Definition~\ref{def:local-guards} be
  associated to every state $q \in Q$. We define a relation
  $\fleq_{\Aa}$ between configurations of $\Aa$ as
  $(q, v) \as (q, v')$ if $v \gq v'$.
\end{definition}

\begin{lemma}\label{lem:a-simulation}
  The relation $\as$ is a simulation on the configurations of timed
  automaton $\Aa$.
\end{lemma}
\begin{proof} Pick two configurations such that $(q, v) \as (q,
  v')$. We need to show the two conditions required for a simulation
  as given in Page~\pageref{ta-simulation}. The first condition
  follows directly from Lemma~\ref{lem:lug-over-time}. For the second
  condition, pick a transition $t := (q, g, R, q_1)$ such that
  $(q, v) \xra{t} (q_1, v_1)$. Since $(q, v) \as (q,v')$, we have
  $v \gq v'$ by definition. Moreover, the guard $g \in \Gg(q)$, by
  Definitions~\ref{def:a-simulation} and \ref{def:local-guards}. This
  implies that $v'$ satisfies the guard $g$ (taking $\d = 0$ in
  Definition~\ref{def:lug-simulation}) thereby giving us a transition
  $(q, v') \xra{t} (q_1, v'_1)$. It remains to show $v_1 \gqone
  v'_1$. This follows directly from the definition of the weakest
  precondition $\wp(\lugone, R)$.
\end{proof}

As pointed before, Definition \ref{def:lug-simulation} gives a
declarative description of the simulation and it is unclear how to
work with it algorithmically, even when the set of constraints $\Gg$
is finite. The main issue is with the $\forall \d$ quantification,
which is not finite. We will first provide a characterization that
brings out the fact that this $\forall \d$ quantification is
irrelevant for diagonal constraints (essentially because value of
$v(x) - v(y)$ does not change with time elapse). Given a set of constraints $\Gg$,
let $\Gg^{-} \incl \Gg$ be the set of non-diagonal constraints in
$\Gg$.

\begin{proposition}
  \label{prop:characterize-g}
  $v \lug v'$ iff $v \lugdf v'$ and for all diagonal constraints
  $\varphi \in \Gg$, if $v \models \varphi$ then $v' \models \varphi$.
\end{proposition}
\begin{proof}
  The left-to-right implication is direct from
  Definition~\ref{def:lug-simulation}. For the right-to-left
  implication, it is sufficient to note that $v(x) - v(y)$ is
  invariant under time elapse, for every pair of clocks $x, y$.
\end{proof}

It now amounts to solving the $\forall \d$ problem for
non-diagonals. It turns out that the $\lu$ simulation achieves this,
almost. Since the $LU$ bounds cannot distinguish between guards with
$<$ and $\le$ comparisons, the $\lu$ simulation does not characterize
$v \lugdf v'$ completely (converse of Lemma~\ref{lem:lu-to-g} does not
hold). Although we are aware of the (rather
technical) modifications to $\lu$ simulation that are needed for this
characterization, we choose to use the existing $\lu$ directly as this
has already been implemented in tools. This gives us a finer
simulation than $v \lugdf v'$. Let us first define the $LU$ bounds
corresponding to a set of general constraints $\Gg$.

\begin{definition}[$LU$-bounds from $\Gg$]
    \label{def:lu-bounds}
    Let $\Gg$ be a finite set
  of constraints.  We define $LU(\Gg)$ to denote the pair of functions
  $L_\Gg$ and $U_\Gg$ defined as follows:
  \begin{align*}
    L_\Gg(x) = \begin{cases}
      -\infty & \text{ if there is no guard of the form $c \lleq x$ in
        $\Gg$} \\
      \max\{ c ~|~ c \lleq x \in \Gg\} & \text{ otherwise }
    \end{cases} \\
    U_\Gg(x) = \begin{cases} -\infty & \text{ if there is no guard of
                                       the form $x \lleq c$ in
                                       $\Gg$} \\
                                     \max\{ c ~|~ x \lleq c \in \Gg\}
                                     & \text{ otherwise }
                                   \end{cases}
  \end{align*}

\end{definition}

\begin{lemma}
  \label{lem:lu-to-g}
  For every set of constraints $\Gg$, $v \lugg v'$ implies
  $v \lugdf v'$.
\end{lemma}
\begin{proof} Suppose $v \lugg v'$.  For a clock $x$, if
  $v(x) = v'(x)$, then trivially the conditions of $v \lugdf v'$ are
  true for constraints involving $x$. If $v'(x) < v(x)$, then by
  Definition~\ref{def:lu-sim}, we have $L_\Gg(x) < v'(x)$. Hence
  $v, v'$ (and thereby $v + \d, v'+\d$) satisfy all constraints of the
  form $c \lleq x$ from $\Gg$, giving the $\lugdf$ condition for
  $c \lleq x$ guards. Analogous reasoning works for the case when
  $v(x) < v'(x)$, giving the $\lugdf$ condition for the remaining
  guards.
\end{proof}

The above observations call for the next definition and subsequent
lemmas.

\begin{definition}[approximating $\lug$]\label{def:as-lu}
  Let $\Gg$ be a finite set of constraints. We define a relation
  $\luglu$ as follows: $v \luglu v'$ if $v \lugg v'$ and for all
  diagonal constraints $\varphi \in \Gg$, if $v \models \varphi$ then
  $v' \models \varphi$. Similarly, define $\aslu$ as
  $(q, v) \aslu (q, v')$ if $v \lugqlu v'$.
\end{definition}

\begin{lemma}
    \label{lem:aslu-is-finite-simulation}
  The relation $\aslu$ is a finite simulation on the configurations of
  $\Aa$.
\end{lemma}
\begin{proof}
  The fact that $\aslu$ is a simulation follows from definitions of
  $\lu$, the state based guards, and the way we have constructed the
  weakest pre-condition over resets. We now focus on showing that it
  is finite. To do this, we will show that for each $q$, the
  simulation $\luglu$ is finite. Recall that a simulation is finite,
  if the function mapping zones $Z$ to down sets $\down{Z}$ has finite
  range.
  From Theorem~\ref{thm:lu-is-correct-diagonal-free}, simulation $\lugg$ is
  finite. Now consider the relation $v \fleq v'$ if for all diagonal
  constraints $\varphi \in \Gg$, if $v \models \varphi$, then
  $v' \models \varphi$. We will now show that this simulation $\fleq$
  is also finite. Let $\Gg^+ \incl \Gg$ be the set of diagonal
  constraints in $\Gg$; now for each subset $S \incl \Gg^+$ consider
  the set
  $\Rpos^{|X|} \cap \Land_{\varphi \in S} \varphi \cap \Land_{\varphi
    \in \Gg^+ \setminus S} \neg \varphi$. These sets partition the
  space of all valuations into atmost $2^{|\Gg^+|}$ classes. The
  downset $\down{Z}$ for a zone with respect to $\fleq$ will be the union of these
  classes that intersect with $Z$. This shows that $\fleq$ is also
  finite.  Since $\lugg$ and $\fleq$ are finite, it follows that
  $\luglu$ is also finite.
\end{proof}

The $\luglu$ relation can be extended to an algorithm for simulation
between zones, which we do in the next section. With this, we get the
following theorem.

\begin{theorem}
  Given a timed automaton $\Aa$, the state based guards of Definition
  \ref{def:local-guards} can be computed. Zone graph enumeration using
  $\aslu$ simulation is a sound, complete and terminating procedure
  for the reachability problem.
\end{theorem}

\section{Algorithm for $Z \lug Z'$}
\label{sec:algorithm-z-lug}

Fix a finite set of guards $\Gg$. Restating the definition of $\lug$
extended to zones: $Z \lug Z'$ if for all $v \in Z$ there exists a
$v' \in Z'$ such that $v \lug v'$. In this section, we will view the
characterization of $\lug$ as in Proposition~\ref{prop:characterize-g}
and give an algorithm to check $Z \lug Z'$ that uses as an oracle a
test $Z \lugdf Z'$. Using the discussion in the previous section, we
later approximate $Z \lugdf Z'$ with $Z \lugg Z'$ in the
implementation. We start with an observation following from
Proposition~\ref{prop:characterize-g}.

\begin{lemma}
  \label{lem:sufficient-condition-for-lug}
  Let $\varphi:= x - y \lleq c$ be a diagonal constraint in $\Gg$.
  Then $Z \lug Z'$ if and only if
  $Z \cap \varphi \lugp Z' \cap \varphi$ and
  $Z \cap \neg \varphi \lugp Z'$ where
  $\Gg' = \Gg \setminus \{ \varphi \}$.

  If $\Gg$ has no diagonal constraints, $Z \lug Z'$ if and only if
  $Z \lugdf Z'$.
\end{lemma}
\begin{proof}
  Suppose $Z \lug Z'$. Since $Z \cap \neg \varphi \subseteq Z$ clearly
  $Z \cap \neg \varphi \lug Z'$ and hence by definition, we get
  $Z \cap \neg \varphi \lugp Z'$.  For the other part let us choose
  $v \in Z$ such that $v \models \varphi$. As $Z \lug Z'$, there
  exists $v' \in Z'$ such that $v \lug v'$. By definition,
  $v' \models \varphi$ since $\varphi \in \Gg$ and
  $v \models \varphi$. Hence, $Z \cap \varphi \lug Z' \cap \varphi$
  and thereby $Z \cap \varphi \lugp Z' \cap \varphi$.

  Conversely, let us suppose $Z \cap \varphi \lugp Z' \cap \varphi$
  and $Z \cap \neg \varphi \lugp Z'$. Choose a $v \in Z$. Either
  $v \in Z \cap \varphi$ or $v \in Z \cap \neg \varphi$.  Suppose
  $v \in Z \cap \varphi$. Then there exists $v' \in Z' \cap \varphi$
  such that $v \lugp v'$. Firstly, note that the set of non-diagonals
  in $\Gg$ and $\Gg'$ are the same, since $\Gg$ and $\Gg'$ differ only
  by the diagonal constraint $\varphi$. Secondly, note that both
  $v \models \varphi$ and $v' \models \varphi$. These two observations
  coupled with $v \lugp v'$ give $v \lug v'$.  Suppose
  $v \in Z \cap \neg \varphi$. Then, there is a $v' \in Z'$ such that
  $v \lugp v'$. Similar to the above case, we already have
  $v \lugdf v'$. Since $v$ does not satisfy $\varphi$, the fact
  $v \lugp v'$ implies the second requirement for $v \lug v'$. Hence
  in both the cases we get $v' \in Z'$ such that $v \lug v'$, proving
  that $Z \lug Z'$.

  The last part of the Lemma follows from Proposition \ref{prop:characterize-g}.
\end{proof}

This leads to the following algorithm consisting of two mutually
recursive procedures. This algorithm is essentially an implementation
of the above lemma, with two optimizations:
\begin{itemize}[nosep]
\item we start with the non-diagonal check in Line 6 of Algorithm 1 -
  if this is already violated, then the algorithm returns false;
\item suppose $Z \lugdf Z'$, the next task is to perform the checks in
  the first statement of
  Lemma~\ref{lem:sufficient-condition-for-lug} - this is done by
  Algorithm 2; note however that when Algorithm 2 is called, we
  already have $Z \lugdf Z'$, hence $Z \cap \neg \varphi \lugdf
  Z'$. Therefore we use an optimization in Line 7 by calling Algorithm
  2 directly (as the check in Line 6 of Algorithm 1 will be
  redundant).
\end{itemize}

\begin{center}
  \begin{minipage}[t]{5cm}
    \begin{algorithm}[H]
      \label{algo:1}
      \vspace{0pt} \SetAlgoCaptionSeparator{} \DontPrintSemicolon
      \SetKwProg{Def}{check}{:}{}
      \SetKwIF{If}{ElseIf}{Else}{if}{:}{else if}{else :}{}
      \SetKwFor{For}{for}{:}{} \SetKwFor{While}{while}{:}{}

      \SetKwData{DiagConstraints}{DiagConstraints}
      \SetKwData{true}{true} \SetKwData{false}{false}

      \SetKwFunction{simulates}{$Z \lug Z'$} \bigskip \bigskip
      \Def{\simulates}{ \If{$Z = \emptyset$}{\Return \true}
        \If{$Z' = \emptyset$}{\Return \false} \If{$Z \not \lugdf Z'$}{
          \Return \false} \Return $Z \lugl Z'$ \; } \bigskip \smallskip
      \smallskip \smallskip
      \caption{}
    \end{algorithm}
  \end{minipage}
  \begin{minipage}[t]{6cm}
    \begin{algorithm}[H]
      \vspace{0pt}
      \label{algo:2}
      \SetAlgoCaptionSeparator{} \DontPrintSemicolon
      \SetKwProg{Def}{check}{:}{}
      \SetKwIF{If}{ElseIf}{Else}{if}{:}{else if}{else :}{}
      \SetKwFor{For}{for}{:}{} \SetKwFor{While}{while}{:}{}

      \SetKwData{DiagConstraints}{DiagConstraints}
      \SetKwData{true}{true} \SetKwData{false}{false}

      \SetKwFunction{simulates}{$Z \lugl Z'$}

      \Def{\simulates}{ \If{$\Gg$ does not contain any diagonal
          constraints}{\Return \true} pick a diagonal constraint
        $\varphi = x - y \lleq c$ from $\Gg$ \; \smallskip
        $\Gg' \longleftarrow \Gg \setminus \{\varphi\}$\; \smallskip
        \If{$Z \cap \neg\varphi \neq \emptyset$}{
            \If {$Z \cap \neg\varphi \not\lugpl Z'$}{\Return \false} }
        \Return $Z \cap \varphi \lugp Z' \cap \varphi$ }
      \caption{}
    \end{algorithm}
  \end{minipage}
\end{center}

As mentioned before, we approximate $Z \lugdf Z'$ with $Z \lugg
Z'$. In this case, each zone operation inside a call can be done in
$\Oo(|X|^2)$~\cite{Zhao:IPL:2005,Herbreteau:IandC:2016}.

\begin{theorem}
    \label{thm:algo-correct-and-terminates}
  When using $Z \lugg Z'$ in the place of $Z \lugdf Z'$, the algorithm terminates in $\Oo(2^d \cdot |X|^2)$
  where $d$ is the number of diagonal guards in $\Gg$.
\end{theorem}

\begin{proof}
  The algorithm makes two recursive calls per diagonal
  constraint. In each call the operations are either $Z \lugg Z'$ or
  an intersection of $Z$ with a single
  constraint. From~\cite{Herbreteau:IandC:2016}, the check $Z \lugg Z'$
  can be done in $\Oo(|X|^2)$. Intersection with a single constraint
  can also be done in $\Oo(|X|^2)$ (similar to approach in
  \cite{Zhao:IPL:2005}). This gives the complexity.
\end{proof}

From a complexity viewpoint, this algorithm is not efficient since it
makes an exponential number of calls in the number of diagonal
constraints (in fact this may not be
avoidable due to Lemma~\ref{lem:hardness}, which follows from the
NP-hardness result in \cite{Gastin:2018:CONCUR}).
Although the above algorithm does involve many calls, the internal operations
involved in each call are simple zone manipulations. Moreover, the
preliminary checks (for instance line 6 of Algorithm 1) cut short the
number of calls. This is visible in our experiments which are very good,
especially with respect to running time, as compared to other
methods. A similar hardness was shown for a different simulation in
\cite{Gastin:2018:CONCUR}, but the implementation there
indeed witnessed the hardness, as the time taken by their algorithm
was unsatisfactory.

\subsection*{Hardness of checking $Z \not\luglu Z'$}
\label{sec:hardness}
In this section we prove that the problem of checking
$Z \not\luglu Z'$ is NP-complete.  In order to do this, we will give a
polynomial time reduction from the problem of checking
$Z \not\lud Z'$, which was shown to be NP-hard
\cite{Gastin:2018:CONCUR}, to our problem. Let us first recall the
$L^dU^d$ bounds function (superscripting with $d$ to distinguish from
the diagonal free $LU$ bounds) and the simulation relation $\lud$ of
\cite{Gastin:2018:CONCUR}.

The bounds $L^d,U^d$ are functions
$L^d:X \times X \rightarrow \mathbb{Z}\cup \{\infty\}$,
$U^d: X \times X \rightarrow \mathbb{Z}\cup\{-\infty\}$ that map every
ordered pair of clocks $x,y$ either to an integer or $\infty$ or
$-\infty$.
It additionally
 satisfies that $L^d(x - 0) = U^d(0 - x) = 0$ for every clock $x$.
Given two valuations $v, v'$ and $L^dU^d$, say $v \lud v'$ if for
every pair of distinct clocks $x, y$ (including a special clock
denoting $0$) the following holds:
\begin{itemize}
\item if $v(x) - v(y) < L^d(x - y)$ then $v'(x) - v'(y) < L^d(x - y)$
\item if $L^d(x - y) \le v(x) - v(y) \le U^d(x - y)$ then
  $v'(x) - v'(y) \le v(x) - v(y)$
\end{itemize}

$Z \lud Z'$ is defined in the usual way: $Z \lud Z'$ if for every
$v \in Z$ there exists $v' \in Z'$ such that $v \lud v'$.
The following two results hold when $Z, Z'$ and the $L^dU^d$ bounds
satisfy \emph{certain} properties.
\begin{lemma}
  \label{lem:integral-valuation}
  \cite{Gastin:2018:CONCUR} Under certain assumptions on $Z, Z'$ and
  $L^dU^d$:
  \begin{enumerate}
  \item $Z \not \lud Z'$ implies there exists an integral valuation
    $v \in Z$ such that for every $v' \in Z'$, we have $v \not\lud v'$
  \item $Z \not \lud Z'$ is NP-hard.
  \end{enumerate}
\end{lemma}

We will use this to show $\lug$ is NP-hard.  Given $L^dU^d$, we
construct $\Gg$ to be the set containing: for every
$x, y \in X \cup \{0\}$ the constraint $x - y < L^d(x - y)$ and the
constraints $x - y \le c$ for every
$c \in [L^d(x - y), U^d(x - y)] \cap \mathbb{Z}$.

\begin{lemma}
  \label{lem:luglu-iff-lud}
  Given $L^dU^d$, let $\Gg$ be constructed as above. Then for all
  $Z, Z'$ satisfying the conditions needed for
  Lemma~\ref{lem:integral-valuation}, $Z \not\luglu Z'$ if and only if
  $Z \not\lud Z'$.
\end{lemma}

\begin{proof}
  ($\Leftarrow$) Assume $Z \not\lud Z'$. Then from Lemma
  \ref{lem:integral-valuation} we get that there exists an integral
  valuation that serves as a witness.  We choose that integral
  valuation and call it $v$. Then there exists some pair of clocks
  $x, y \in X \cup \{0\}$ such that for every $v' \in Z'$, one of the
  following two cases holds:
  \begin{itemize}
  \item $v'(x) - v'(y) > v(x) - v(y)$ and
    $L^d(x - y) \le v(x) - v(y) \le U^d(x - y)$. In this case, from
    the construction of $\Gg$, we get that $x - y \le c \in \Gg$ where
    $c = v(x) - v(y)$.  Therefore, $v \models x - y \le c$ whereas
    $v' \not\models x - y \le c$.
  \item $v'(x) - v'(y) \ge L^d(x - y)$ and $v(x) - v(y) < L^d(x - y)$.
    Then, $x - y < L^d(x - y) \in \Gg$, and
    $v \models x - y < L^d(x - y)$ but
    $v' \not\models x - y < L^d(x - y)$.
  \end{itemize}
  This proves that, $Z \not\luglu Z'$.

  ($\Rightarrow$) Assume $Z \not\luglu Z'$. Then either
  $Z \not\lugg Z'$ or there exists a diagonal constraint
  $x - y \lleq c \in \Gg$ such that $v \models x - y \lleq c$ for some
  $v \in Z$, but for every $v' \in Z'$, $v' \not\models x - y \lleq c$.
  In the second case, from the construction of $\Gg$ we get that
  $c \in [L^d(x - y), U^d(x - y)]$.  Hence for all $v' \in Z'$,
  $v'(x) - v'(y) > v(x) - v(y)$ but
  $L^d(x - y) \le v(x) - v(y) \le U^d(x - y)$. Therefore
  $Z \not\lud Z'$.  On the other hand, if $Z \not\lugg Z'$ there is
  $v \in Z$ such that for every $v' \in Z'$ one of the following two cases
  happen:
  \begin{itemize}
  \item $v'(x) < v(x)$ and $v'(x) \le L(x)$. For $\lud$ to hold, we
     require that either both $0- v'(x)$ and $0-v(x)$ are $< L^d(0
    - x)$ or both are $> U^d(0 - x)$. Latter cannot happen as $U^d(0 -
    x) = 0$. Former cannot happen as
    $L(x) = -L^d(0 - x)$.
  \item $v(x) < v'(x)$ and $v(x) \le U(x)$. Use similar argument as
    previous case.
  \end{itemize}
  This completes the proof.
\end{proof}

Before proceeding to prove NP-completeness, we present a new
characterization of the simulation relation $\luglu$ that will help in
proving the upper bound.

\begin{lemma}
  \label{lem:characterizing-luglu}
  $Z \not \luglu Z'$ iff there exists a subset $S$ of diagonal
  constraints from $\Gg$ such that
  $(Z \cap \Land_{\varphi \in S} \varphi) ~\not \lugg~ (Z' \cap
  \Land_{\varphi \in S}\varphi)$.
\end{lemma}
\begin{proof}
  Suppose $Z \not \luglu Z'$. There exists $v \in Z$ such that
  $v \not \luglu v'$ for every $v' \in Z'$. Let $S$ be the set of
  diagonal constraints $\varphi \in \Gg$ such that
  $v \models \varphi$. Then,
  $v \in Z \cap \Land_{\varphi \in S} \varphi$ and hence
  $(Z \cap \Land_{\varphi \in S} \varphi) ~\not \lugg~ (Z' \cap
  \Land_{\varphi \in S} \varphi)$, as otherwise there is a $v'$ such
  that $v \lugg v'$ and $v'$ satisfies all diagonal constraints of
  $\Gg$ that $v$ satisfies.

  Suppose there exists $S$ such that
  $(Z \cap \Land_{\varphi \in S} \varphi) ~\not \lugg~ (Z' \cap
  \Land_{\varphi \in S} \varphi)$. Then, clearly
  $(Z \cap \Land_{\varphi \in S} \varphi)$ is non-empty and for
  valuations $v$ in this set, there is no $v' \in Z'$ with
  $v \luglu v'$.
\end{proof}

\begin{theorem}
    \label{lem:hardness}
  Given two zones $Z, Z'$ and a set of constraints $\Gg$, checking
  $Z \not\luglu Z'$ is NP-complete.
\end{theorem}

\begin{proof}
  From Lemma \ref{lem:characterizing-luglu} it follows that a set
  $S~(\subseteq \Gg)$ of diagonal constraints serves as a witness for
  $Z \not\luglu Z'$.  Given this set $S$,
  $(Z \cap \Land_{\varphi \in S} \varphi) \not \lugg (Z' \cap
  \Land_{\varphi \in S}\varphi)$ can be checked in quadratic time
\cite{Herbreteau:IandC:2016}.  Hence the problem of checking
  $Z \not\luglu Z'$ is in NP.

  For the hardness part, we reduce the problem of checking
  $Z \not\lud Z'$ to $Z \not\luglu Z'$.  Given, $Z, Z'$ and the
  bounds $L^d, U^d$ with the required conditions for
  Lemma~\ref{lem:integral-valuation}, we first construct $\Gg$ as we have described
  before.  Lemma \ref{lem:luglu-iff-lud} shows that $Z \not\lud Z'$
  can be aswered by checking $Z \not\luglu Z'$.  Hence, Lemma
  \ref{lem:integral-valuation} implies that checking $Z \not\luglu Z'$ is
  NP-hard.
\end{proof}

\section{Simulations for updateable timed automata}
\label{sec:updates}

In the timed automata considered so far, clocks are allowed to be reset to $0$
along transitions.  We consider in this section more sophisticated
transformations to clocks in transitions.  These are called \emph{updates}.  An
update $up: \Rpos^{|X|} \mapsto \mathbb{R}^{|X|}$ is a function mapping
non-negative $|X|$-dimensional reals (valuations) $v$ to general
$|X|$-dimensional reals (which may apriori not be valuations as the coordinates
may be negative).  The syntax of the update function $up$ is given by a set of
atomic updates $up_x$ to each $x \in X$, which are of the form $x:= c$ or $x :=
y + d$ where $c\in \Nat$, $d \in \Int$ and $y \in X$ (possibly equal
to $x$).  Note that we want $d$ to be an
integer, since we allow for decrementing clocks, and on the other hand
$c \in \Nat$ since we have non-negative clocks.  Given a valuation $v$ and an
update $up$, the valuation $up(v)$ is:
\begin{align*}
  up(v)(x) := \begin{cases}
    c & \text{ if  $up_x$ is $x := c$ }\\
    v(y) + d & \text{ if $up_x$ is $x:= y + d$}
  \end{cases}
\end{align*}
Note that in general, due to the presence of updates $x:= y +d$, the
update $up(v)$ may not yield a clock valuation. However, when it does
give a valuation, it can be used as a transformation in timed automata
transitions. We say $up(v) \ge 0$ if $up(v)(x) \ge 0$ for all clocks
$x\in X$.

An \emph{updateable timed automaton} $\Aa = (Q, X, q_0, T, F)$ is an
extension of a classic timed automaton with transitions of the form
$(q, g, up, q')$ where $up$ is an update. Semantics extend in the
natural way: delay transitions remain the same, and for action
transitions $t:= (q, g, up, q')$ we have $(q,v) \xra{t} (q', v')$ if
$v \models g$, $up(v) \ge 0$, and $v' = up(v)$. We allow the
transition only if the update results in a valuation. The reachability
problem for these automata is known to be undecidable in
general~\cite{Bouyer:2004:Updateable}. Various subclasses with
decidable reachability have been discussed in the same
paper. Decidability proofs in~\cite{Bouyer:2004:Updateable} take the
following flavour, for a given automaton $\Aa$: (1) divide the space
of all valuations into a finite number of equivalence classes called
\emph{regions} (2) to build the parameters for the equivalence, derive a set of
diophantine equations from the guards of $\Aa$; if they have a
solution then construct the quotient graph of the equivalence (called
region graph) parameterized by the obtained solution and check reachability on it; if the equations have no
solution, output that reachability for $\Aa$ cannot be
answered. Sufficient conditions on the nature of the updates that give
a solution to the diophantine equations have been tabulated
in~\cite{Bouyer:2004:Updateable}. When the automaton is diagonal-free,
the ``region-equivalence'' can be used to build an extrapolation
operation which in turn can be used in a reachability algorithm with
zones. When the automaton contains diagonals, the region-equivalence
is used to only build a region graph - no effective zone based approach has
been studied.

We use a similar idea, but we have two fundamental differences: (1) we
want to obtain reachability through the use of simulations on zones,
and (2) we build equations over sets of guards as in
Definition~\ref{def:local-guards}. The advantage of this approach is
that this allows the use of coarser simulations over zones.
Even for
automata with diagonal constraints and updates, we get a zone based
algorithm, instead of resorting to regions which are not efficient in practice.

The notion of simulations as in Page \pageref{def:lu-sim} remains the
same, now using the semantics of transitions with updates. We will
re-use the simulation relation $\lug$. We need to extend Definition
\ref{def:wp-resets} to incorporate updates. We do this below. Here is
a notation: for an update function $up$, we write $up(x)$ to be $c$ if
$up_x$ is $x := c$, and $up(x)$ to be $y + c$ if $up_x$ is $x:= y+c$.

\begin{definition}[weakest pre-condition of $\lug$ over updates]\label{def:wp-up}
  \mbox{ }

  Let $up$ be an update.

  For a constraint $\varphi$ of the form $x \lleq c$ or $c \lleq x$, we define
  $\wp(\lup, up)$ to be respectively $\{up(x) \lleq c\}$ or $\{c \lleq
  up(x)\}$ if these constraints
  are of the form $z \lleq d$ or $d \lleq z$ with $z \in X$ and $d \ge 0$, otherwise
  $\wp(\lup, up)$ is empty.

  For a constraint $\varphi: x - y \lleq c$, we define $\wp(\lup, up)$ to be
  $\{up(x) - up(y) \lleq c\}$ if this constraint is either a diagonal using different
  clocks, or it is of the form $z \lleq d$ or $d \lleq z$ with $d \ge 0$, otherwise
  $\wp(\lup, up)$ is empty.

  For a set of guards $\Gg$, we define
  $\wp(\lug, up) := \bigcup_{\varphi \in \Gg} \wp(\lup, up)$.
\end{definition}

Some examples: $\wp(x \le 5, x := x + 10)$ is empty, since $up(x)$
is $x + 10$, and the guard $x + 10 \le 5$ is not satisfiable;
$\wp(x \le 5, x := x - 10)$ is $x\leq15$,
$\wp(x \le 5, x := c)$ is empty,
$\wp(x - y \le 5, \langle x:= z_1, y:= z_2 + 10 \rangle)$ will be $z_1 - (z_2 + 10)
\le 5$, giving the constraint $z_1 - z_2 \le 15$,
$\wp(x - y \le 5, \langle x:= z+c_1, y:= z + c_2 \rangle)$ is empty,
$\wp(x - y \le 5, \langle x:= c_1, y:= z+c_2 \rangle)$ is
$c=c_1-5-c_2\leq z$ if $c\geq0$ and is empty otherwise.

\begin{definition}[State based guards]
  \label{def:local-guards-update}
  Let $\Aa = (Q, X, q_0, T, F)$ be an updateable timed automaton. We associate a
  set of constraints $\Gg(q)$ for each state $q \in Q$, which is the least
  set of constraints (for the coordinate-wise subset inclusion order) such
  that for every transition $(q, g, up, q_1)$: the guard $g$ and the
  set $\wp(\lugone, up)$ are present in $\Gg(q)$, and in addition
  constraints that allow the update to happen are also present
  in $\Gg$. The last condition is
  given by the weakest precondition of the set of constraints $\{x \ge
  0~|~ x \in X\}$. Overall, $\{\Gg(q)\}_{q \in Q}$ is the least
  solution to the following set of equations, for each $q \in Q$:
  \begin{align*}
    \Gg(q) = \bigcup_{(q, g, up, q_1) \in T} \left(~\{g\} ~\cup~
    \wp(\sqsubseteq_{ \{ x \ge 0 \mid x \in X \}}, up) ~\cup~ \wp(\lugone, up)~\right)
  \end{align*}
  The least solution $\{\Gg(q)\}_{q \in Q}$ is
  said to be finite if each $\Gg(q)$ is a finite set of constraints.
\end{definition}

In contrast to the simple reset case, the above set of equations may
not have a finite solution. Consider a
self-looping transition: $(q, x \lleq c, x:= x - 1, q)$. We require
$x \lleq c \in \Gg(q)$. Now, $\wp(x \lleq c, x:= x -1)$ is
$ x \lleq c+1$ which should be in $\Gg(q)$ according to the above
equation. Continuing this process, we need to add $x \lleq d$ for
every natural number $d \ge c$. Indeed this is consistent with the
undecidability of reachability when subtraction updates are
allowed. We deal with the subject of finite solutions to the above
equations later in this section. On the other hand,
when the above system does have a solution with finite $\Gg(q)$ at
every $q$, we can use the $\Aa$ simulation of
Definition~\ref{def:a-simulation} and its approximation $\aslu$ to get an algorithm.

\begin{proposition}
  \label{prop:simulation-for-uta}
  Let $\Aa = (Q, X, q_0, T, F)$ be an updateable timed automaton. Let
  $\{\Gg(q)\}_{q \in Q}$ with be the
  least solution to the equations given in
  Definition~\ref{def:local-guards-update}. Then, the relation $\as$ is a
  simulation on the configurations of $\Aa$.
\end{proposition}
\begin{proof}
  Pick two configurations such that $(q, v) \as (q, v')$. We need to
  show the two conditions required for a simulation as given in
  Page~\pageref{ta-simulation}. The proof follows along the lines of
  Lemma~\ref{lem:a-simulation}, except for one change. Consider a
  transition $t:(q, g, up, q_1)$ and let $(q, v) \xra{t} (q_1,
  v_1)$. This implies that $v \models g$ and $up(v)$ is defined, that
  is $up(v)(x) \ge 0$ for all $x$.  Since
  $g \in \Gg(q)$, we get that $v' \models g$.  Moreover, as $\Gg(q)$
  contains $\wp(\sqsubseteq_{ \{ x \ge 0 \mid x \in X \}}, up)$, and
  $v \gq v'$, we get that $up(v')$ is defined.  Now, let
  $(q, v') \xra{t} (q_1, v'_1)$.

  To show $v_1 \lugone v'_1$, where $v_1 = up(v)$ and $v'_1 =
  up(v')$. Pick a constraint $\varphi\in\Gg(q_1)$ such that $v \models
  \varphi$. Suppose $\varphi$ is of the form $x - y \lleq c$. If
  both $up(x)$ and $up(y)$ contain the same clock $z$, then in the
  difference $x - y$ the role of $z$ gets cancelled, and we get that
  $v_1(x) - v_1(y)$ equals $v'_1(x) - v'_1(y)$, and furthermore this
  difference is invariant over time elapse. Hence the $\lugone$
  property is true for $\varphi$. The next case is when $up(x)$ and
  $up(y)$ do not contain the same clock. Then, $up(x) - up(y) \lleq c$
  is a guard (diagonal or non-diagonal) which according to
  Definition~\ref{def:wp-up} is present in $\Gg$. Since $v \lug v'$,
  we get the $\lugone$ condition for $\varphi$. Similar reasoning
  follows when $\varphi$ is a non-diagonal. One however needs to use
  the fact that since $v \models \varphi$, if $\varphi$ is of the form
  $x \lleq c$ then $up(x) \lleq c$ is a valid guard.
\end{proof}

\begin{lemma}
For an updateable timed automaton $\Aa$, assume the least solution
$\{\Gg(q)\}_{q \in Q}$ to  the state-based guards equations is
finite. Then the relation $\aslu$ is a finite simulation on the
configurations of $\Aa$.
\end{lemma}

\subsubsection*{Finite solution to the state-based guards equations}
\label{sec:finite-solut-state}
The least solution to the equations of Definition
\ref{def:local-guards-update} can be obtained by a standard Kleene
iteration for fixed points computation. For each $i \ge 0$ and each
state $q$, define:
\begin{align*}
  \Gg^{0}(q) &= \bigcup_{(q, g, up, q') \in T}
  \{g\}\cup\wp(\sqsubseteq_{ \{ x \ge 0 \mid x \in X \}}, up)
  \\
  \Gg^{i+1}(q) &= \bigcup_{(q, g, up, q') \in T}
  \Gg^{i}(q)\cup \wp(\lugpar{\Gg^{i}(q')}, up)
\end{align*}
 The iteration stabilizes when there exists a
$k$ satisfying $\Gg^{k+1}(q) = \Gg^k(q)$ for all $q$.  At stabilization, the
values $\Gg^k(q)$ satisfy the equations of
Definition~\ref{def:local-guards-update}, and give the required $\Gg(q)$.
However, as we mentioned earlier, this iteration might not stabilize at any $k$.
We will now develop some observations that will help detect after finitely many
steps if the iteration will stabilize or not.

Suppose we colour the set $\Gg^{i+1}(q)$ to \emph{red} if either there
exists a diagonal constraint
$x - y \lleq c \in \Gg^{i+1}(q) \setminus \Gg^{i}(q)$ (a new diagonal
is added) or there exists a non-diagonal constraint $x \lleq c$ or
$c \lleq x$ in $\Gg^{i+1}(q) \setminus \Gg^{i}(q)$ such that the
constant $c$ is strictly bigger than $c'$ for respectively every
non-diagonal $x \lleq c'$ or $c' \lleq x$ in $\Gg^{i}(q)$ (a
non-diagonal with a bigger constant is added).  If this condition is
not applicable, we colour the set $\Gg^{i+1}(q)$ \emph{green}. The
next observations say that the iteration terminates iff we reach a
stage where all sets are green. Intuitively, once we reach green,
the only constraints that can be added are non-diagonals having
smaller (non-negative) constants and hence the procedure terminates.

\begin{lemma}
  \label{lem:green-green}
  Let $i > 0$. If $\Gg^{i}(q)$ is green for all $q$, then
  $\Gg^{i+1}(q)$ is green for all $q$.
\end{lemma}

\begin{proof}
  Pick an $i>0$ such that $\Gg^i(q)$ is green for all $q$. Suppose there
  is a state $p$ such that $\Gg^{i+1}(p)$ is red. By definition, there
  is a new constraint $\varphi \in \Gg^{i+1}(p) \setminus
  \Gg^i(p)$. This means that there is a transition $(p, g, up, p_1)$
  and a constraint $\psi \in \Gg^i(p_1) \setminus \Gg^{i-1}(p_1)$
  (which is new at the $i^{th}$ stage) such that $\varphi$ is
  $\wp(\sqsubseteq_{\psi}, up)$.

  If $\varphi$ is a diagonal constraint, then $\psi$ is also a
  diagonal constraint.  This contradicts $\Gg^i(p_1)$ being green and
  hence $\varphi$ cannot be a diagonal constraint. Suppose $\varphi$
  is a non-diagonal $x \lleq c$. Then, $\psi$ is of the form
  $z \lleq d$ and the update $up$ is of the form $z := x +
  c_1$. Therefore $\varphi$ is $x \lleq d - c_1$. If $\Gg^{i-1}(p_1)$
  had a constraint $z \lleq d'$ with $d' \ge d$, then a constraint
  $x \lleq d' - c_1$ would have been added to $\Gg^i(p)$. We know that
  is not the case as $\Gg^{i+1}(p)$ has been coloured red due to
  $\varphi$, and hence $d - c_1 > d' - c_1$. This gives $d > d'$ and
  hence $\Gg^i(p_1)$ should be coloured red due to the added
  constraint $z \lleq d$. Yet again, this contradicts that
  $\Gg^i(p_1)$ is green. The case when $\varphi$ is of the form
  $c \lleq x$ is similar.
\end{proof}

\begin{lemma}
  \label{lem:K-red-no-green}
  Let $K = 1+|Q|\cdot|X|\cdot(|X|+1)$. If there is a state $p$ such that
  $\Gg^K(p)$ is red, then there is no $i$ such that $\Gg^i(q)$ is
  green for all $q$.
\end{lemma}

\begin{proof}
  Suppose there is a state $p_K$ such that $\Gg^{K}(p_K)$ is red.  This is due
  to some constraint $\varphi_K\in\Gg^{K}(p_K)\setminus\Gg^{K-1}(p_K)$.  Hence
  there is a transition $t_K=(p_{K},g_K,up_K,p_{K-1})\in T$ and a constraint
  $\varphi_{K-1}\in\Gg^{K-1}(p_{K-1})\setminus\Gg^{K-2}(p_{K-1})$ such that
  $\varphi_K\in\wp(\lugpar{\varphi_{K-1}},up_{K})$.
  We claim that $\Gg^{K-1}(p_{K-1})$ is red due to $\varphi_{K-1}$.  This is
  clear if $\varphi_{K-1}$ is a diagonal constraint.
  Assume $\varphi_{K-1}=x\lleq c$.  Then $\varphi_{K}=y\lleq d$ for some
  $d\geq0$.  Assume there is some $x\lleq c'\in\Gg^{K-1}(p_{K-1})$ with $c'\geq
  c$.  Then we have $y\lleq d'\in\wp(\lugpar{x\lleq
  c'},up_{K})\subseteq\Gg^{K}(p_K)$ with $d'\geq d$, a contradiction with the
  fact that $\Gg^{K}(p_K)$ is red due to $\varphi_K$.  The proof for the case
  $\varphi_{K-1}=c\lleq x$ is similar.

  Iterating this construction we obtain a sequence
  $(t_j=(p_{j},g_j,up_j,p_{j-1}))_{2\leq j\leq K}$ of transitions and a sequence
  of constraints $(\varphi_j)_{1\leq j\leq K}$ such that
  $\varphi_j\in\wp(\lugpar{\varphi_{j-1}},up_{j})$ for $2\leq j\leq K$ and
  $\Gg^j(p_j)$ is red due to the addition of $\varphi_j$ for $1\leq i\leq K$.
  Since $K = 1+|Q|\cdot|X|\cdot(|X|+1)$ there exist two indices $a$ and $b$ with
  $b > a$ such that $p_a = p_b=p$ and both $\varphi_a$ and $\varphi_b$ have the
  same ``support'', that is, either $\varphi_a:= x - y \lleq c$ and $\varphi_b
  := x - y \lleq c'$, or $\varphi_a:=x \lleq c$ and $\varphi_b:=x \lleq c'$, or
  $\varphi_a:=c \lleq x$ and $\varphi_b:=c' \lleq x$.

  Suppose $\varphi_a$ and $\varphi_b$ are non-diagonals.  Since $\Gg^b(p)$ is
  red due to $\varphi_b$ and $\varphi_a\in\Gg^a(p)\subseteq\Gg^b(p)$, we deduce
  that $c' > c$.  Since $p_a=p_b=p$, the weakest pre-condition computations for
  the sequence of updates $up_{a+1},\ldots,up_{b}$ can be applied starting from
  $\varphi_b\in\Gg^b(p_b)$ and results in a constraint $x\lleq c''$ or $c''\lleq
  x$ with $c''-c'>c'-c>0$.  This can be iterated infinitely often to give red
  sets at each stage.

  Assume now that $\varphi_a:= x-y\lleq c$ and $\varphi_b:= x-y\lleq c'$.
  Since $\varphi_b$ is a new diagonal constraint in $\Gg^b(p)\supseteq\Gg^a(p)$
  we get $c'\neq c$. As above, the weakest pre-condition computations for
  the sequence of updates $up_{a+1},\ldots,up_{b}$ can be applied starting from
  $\varphi_b\in\Gg^b(p_b)$ and results in a constraint $x-y\lleq c''$ with
  $c''-c'=c'-c\neq0$. This can be iterated infinitely often giving infinitely
  many new diagonal constraints.

  Both cases imply that there can be no $i$ where all $\Gg^i(q)$ are green.
\end{proof}

\begin{proposition}
    \label{prop:least-solution-is-finite-iff}
  The least solution of the local constraint equations for an
  updateable timed automaton is finite iff $\Gg^K(q)$ is green for all
  $q$ and where $K = 1+|Q|\cdot|X|\cdot(|X|+1)$.
\end{proposition}
\begin{proof}
  If $\Gg^K(q)$ is green for all $q$, then it is green for all
  $i \ge K$ (from Lemma \ref{lem:green-green}). By definition of
  green, there cannot be infinitely many new constraints added after
  stage $K$. Hence the iteration stabilizes, giving a finite solution.

  If $\Gg^K(p)$ is red for some $p$, Lemma~\ref{lem:K-red-no-green}
  gives that the iteration does not terminate and hence there is no
  finite solution.
\end{proof}

\begin{theorem}
  Let $\Aa$ be an updateable timed automaton. It is decidable whether
  the equations in Definition~\ref{def:local-guards-update} have a
  finite solution. When these equations do have a finite solution,
  zone graph enumeration using $\aslu$ is a sound, complete and
  terminating procedure for the reachability problem.
\end{theorem}

All decidable classes of \cite{Bouyer:2004:Updateable} can be shown
decidable with our approach, by showing stabilization of the $\Gg(q)$
computation.

\begin{lemma}
    \label{lem:reachability-decidable-in-uta}
  Reachability is decidable in updateable timed  automata where:
  guards are non-diagonals and updates are of the form $x:= c$,
  $x:=y$,
      $x:= y + c$ where $c \ge 0$ or, guards include diagonal constraints and updates are of the
      form $x:= c$, $x := y$.
  \end{lemma}
  \begin{proof}
    With the above combination of guards and updates, each constraint added in
    the $\Gg(q)$ computation apart from the guards of the automaton is either a
    diagonal constraint with a constant appearing in the guards of the
    automaton, or is a non-diagonal constraint with a constant $c\leq2M$ where
    $M=\max\{|d|\mid d \text{ occurs in a guard or an update of the
    automaton}\}$.  Hence the $\Gg(q)$ computation terminates.
  \end{proof}

\section{Experiments}
\label{sec:experiments}
\begin{sidewaystable}[ph!]
  \begin{center}
    \begin{tabular}{|m{6em}|>{\centering\arraybackslash}m{2em}||>{\raggedleft\arraybackslash}m{3.5em}|>{\centering\arraybackslash}m{3.5em}||>{\raggedleft\arraybackslash}m{3.5em}|>{\centering\arraybackslash}m{3.5em}||>{\raggedleft\arraybackslash}m{3.5em}|>{\centering\arraybackslash}m{3.5em}||>{\raggedleft\arraybackslash}m{3.5em}|>{\centering\arraybackslash}m{3.5em}|}
      \hline
      & & \multicolumn{4}{c||}{$\Aa$ : contains diagonals} & \multicolumn{4}{c|}{$\Aa_{df}$ : diagonal-free \emph{euivalent} of $\Aa$} \\
      \hline
      \hline
      & & \multicolumn{2}{c||}{TChecker with $\luglu$} & \multicolumn{2}{c||}{UPPAAL} & \multicolumn{2}{c||}{ UPPAAL} & \multicolumn{2}{c|}{ TChecker} \\
      \hline
      Model & $\#\mathcal{D}$ &Time & Nodes count & Time & Nodes count & Time & Nodes count & Time & Nodes count \\
      \hline
      \hline
       Cex 1 & 2 & 0.004 & 7 & 0.001 & 26 & 0.001 & 17 & 0.006 & 17 \\
      \hline
      Cex 2 & 4 & 0.047 & 241 & 0.026 & 2180 & 0.005 & 1039 & 0.067 & 1039 \\
      \hline
      Cex 3 & 6 & 7.399 & 7111 & 111.168 & 182394 & 1.028 & 60982 & 40.092 & 60982 \\
      \hline
      Cex 4 & 8 & 857.662 & 185209 & timeout & - & 734.543 & 3447119 & timeout & - \\
      \hline

      \hline
      Fischer 3 & 3 & 0.007 & 104 & 0.087 & 4272 & 0.001 & 268 & 0.013 & 268 \\
      \hline
      Fischer 4 & 4 & 0.032 & 452 & 307.836 & 357687 & 0.009 & 1815 & 0.100 & 1815 \\
      \hline
      Fischer 5 & 5 & 0.257 & 1842 & timeout & - & 0.116 & 12511 & 1.856 & 12511 \\
      \hline
      Fischer 7 & 7 & 15.032 & 26812 & timeout & - & 174.560 & 693603 & timeout & - \\
      \hline

    \hline
    Job Shop 3 & 12 & 0.420 & 278 & 23.093 & 31711 & 0.003 & 845 & 0.312 & 845 \\
    \hline
    Job Shop 5 & 20 & 285.421 & 10592 & timeout & - & 4.633 & 179607 & 150.811 & 179607 \\
    \hline
    Job Shop 7 & 28 & timeout & - & timeout & - & timeout & - & timeout & - \\
    \hline

      \hline
      \hline
      Job Shop 3 & 12 & 0.019 & 38 & 22.435 & 31607 & 0.004 & 839 & 0.013 & 29 \\
      \hline
      Job Shop 5 & 20 & 0.279 & 98 & timeout & - & 4.552 & 179597 & 0.012 & 67 \\
      \hline
      Job Shop 7 & 28 & 1.754 & 192 & timeout & - & timeout & - & 0.036 & 121 \\
      \hline
      Job Shop 9 & 36 & 7.040 & 318 & timeout  & - & timeout & - & 0.056 & 191 \\
      \hline
    \end{tabular}
  \end{center}
  \caption{Experiments: the column $\# \mathcal{D}$ gives the number of diagonal
    constraints. There are four methods reported in the table.
    First algorithm (TChecker with $\luglu$) is our implementation of
    simulation based reachability using state based guards and
    $\luglu$ simulation; second algorithm (UPPAAL) is the result of running the
    model ($\Aa$) with diagonal constraints in UPPAAL; third algorithm (UPPAAL) gives the
    result of running an equivalent diagonal-free automaton ($\Aa_{df}$) in UPPAAL and the fourth algorithm (TChecker) gives the result of running the same (diagonal-free) model ($\Aa_{df}$) in TChecker. Experiments were run on a
    MacBook Pro laptop with 2.3 GHz Intel core i5 processor, and 8 GB
    RAM. Time is reported in seconds. We set a timeout of 15 minutes.}
  \label{tab:experiments}
\end{sidewaystable}

We have implemented the reachability algorithm for timed automata with
diagonal constraints (and only resets as updates) based on the
simulation approach (Page \pageref{algo:simulation-reachability})
using the $\aslu$ simulation (Definition \ref{def:as-lu}) for pruning
zones. The algorithm for $Z \luglu Z'$ comes from Section
\ref{sec:algorithm-z-lug}.
Our implementation is over a prototype tool
TChecker \cite{tchecker}.
The core reachability algorithm based on the
simulation approach is already implemented in TChecker for timed
automata without diagonal constraints. We have added the state based
guards computation, and the $\luglu$ simulation algorithm.
In the rest of this section, we explain the timed automata models and
the experimental results (Table~\ref{tab:experiments}).
Experiments are reported in Table~\ref{tab:experiments}. We take
model \emph{Cex} from~\cite{Bouyer:2003:STACS,Reynier:toolreport} and \emph{Fischer}
from~\cite{Reynier:toolreport}. We are not aware of any other
``standard'' benchmarks containing diagonal constraints. In addition
to these two models, we introduce the following new benchmark.

\paragraph*{Job-shop scheduling with dependent tasks.} This is
an extension of the job-shop scheduling using (diagonal-free) timed
automata~\cite{Maler:2006:Scheduling}. The benchmark models the
following situation: there are $n$ jobs $J_1, J_2, \dots, J_n$ and $k$
machines $m_1, m_2, \dots, m_k$. Each job $J_i$ consists of two tasks
given by triples $T_i^1:= (m_i^1, a_i^1, b_i^1)$ and
$T_i^2= (m_i^1, a_i^1, b_i^1)$ where $m, a, b$ (with appropriate
indices) denote respectively the machine on which the task needs
to be executed and $[a, b]$ is an interval such that the time needed
to execute the task is in this interval. In addition to this, each job has
an overall deadline $D_i$. The question now is if the tasks can be
scheduled so that all jobs finish within their deadline. This problem
can be modeled using diagonal-free automata. We now add a further
constraint (excluding the indices for clarity): for a job $J$ with
tasks $T^1$ and $T^2$, let $t^1$ and $t^2$ be the execution times. We
know that $t^1 \in [a^1, b^1]$ and $t_2 \in [a^2, b^2]$. Now, we
have a dependency: if $t^1 \ge c^1$, then $t^2 \le c^2$ and if
$t^1 \le d^1$ then $t^2 \ge d^2$ for suitable constants $c^1, c^2,
d^1, d^2$. This says that if the first task is
executed for a longer time, the second task needs a shorter time to
finish and vice versa. This kind of a dependency has a natural
modeling using diagonal constraints. We illustrate in
Figure~\ref{fig:job} a part of the UPPAAL-style automaton for a job
$J$ capturing the timing requirements - we make use of the feature of
committed locations in UPPAAL~\cite{Bengtsson:Springer:2004}. Time is not allowed to elapse in this state, and in a
product the components with committed locations execute first. This
modeling template of the job scheduling problem with dependent tasks
can easily be extended when jobs have more tasks and there are similar
dependencies between them. In addition to this automaton, we model the
mutual exclusion between machines (by use of extra boolean variables): each machine can have atmost
one task at a time. We vary the number of jobs and get different timed
automata with diagonal constraints. Note that although this model is
acyclic, there are ``cross simulations'' which help in pruning the
search: the same state $q$ can be reached by multiple paths and
simulations help in cutting out new searches.

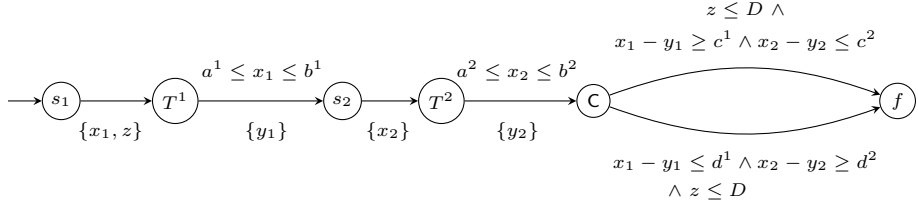
\begin{figure}[t]
  \centering
  \begin{tikzpicture}[state/.style={draw, circle, inner sep=2pt,
      minimum size = 4mm}]
    \begin{scope}[every node/.style={state}]
      \node (0) at (-1.5,0) {\scriptsize $s_1$}; \node (1) at (0,0)
      {\scriptsize $T^1$}; \node (11) at (2.2,0) {\scriptsize $s_2$}; \node (2) at (3.5,0) {\scriptsize $T^2$};
      \node (3) at (5.5,0) {\scriptsize $\mathsf{C}$}; \node (4) at
      (9.5,0) {\scriptsize $f$};
    \end{scope}
    \begin{scope}[->, >=stealth]
      \draw (-2.2,0) to (0); \draw (0) to (1); \draw (1) to (11); \draw (11) to
      (2); \draw (2) to (3); \draw (3) to [bend left=20] (4); \draw (3) to [bend
      right=20] (4);
    \end{scope}
    \node at (-0.85, -0.4) {\scriptsize $\{x_1, z\}$}; \node at (1.15,
    0.4) {\scriptsize $a^1 \le x_1 \le b^1$}; \node at (1.2,
    -0.4) {\scriptsize $\{y_1\}$}; \node at (2.8, -0.4)
    {\scriptsize $\{x_2\}$}; \node at (4.5, 0.4) {\scriptsize
      $a^2 \le x_2 \le b^2$}; \node at (4.5, -0.4) {\scriptsize
      $\{y_2\}$}; \node at (7.5, 0.8) {\scriptsize
      $x_1 - y_1 \ge c^1 \land x_2 - y_2 \le c^2$}; \node at (7.5,
    1.2) {\scriptsize $z \le D~\land$ }; \node at (7.5, -0.8)
    {\scriptsize $x_1 - y_1 \le d^1 \land x_2 - y_2 \ge d^2 $}; \node
    at (7, -1.2) {\scriptsize $\land~z \le D $ };
  \end{tikzpicture}
  \caption{UPPAAL model capturing the timing requirements in a
    scheduling problem. The state marked $\mathsf{C}$ is called a
    committed location, which is a convenient modeling feature in
    UPPAAL. Time is not allowed to elapse in such a state.}
  \label{fig:job}
\end{figure}
Each model considered
above is a product of a number of $k$ timed automata. In the table we
write the name of the model and the number $k$ of automata involved in
the product. Along with this, we report the number of diagonal
constraints in each of them.

\medskip\noindent
\emph{Experimental results.} We considered four algorithms as
mentioned in the caption of Table~\ref{tab:experiments}. Under each
algorithm, we report on the number of zones enumerated and the time
taken. The first algorithm gives a huge gain over the second algorithm
(upto four orders of magnitude in the number of nodes, and even
better for time) and gives a less marked, but still significant, gain
over the third and fourth algorithms. We provide a brief explanation of this
phenomenon. The performance of the reachability algorithm is dependent
on three factors:
\begin{itemize}[nosep]
\item the parameters on which the extrapolation or simulation is based
  on: the maximum constant based $M$-simulations which use the maximum
  constant appearing in the guards, versus the $LU$-simulations which
  make a distinction between lower bound guards $c \lleq x$ and upper
  bound guards $x \lleq c$ (refer to \cite{Behrmann:STTT:2006} for the
  exact definitions of extrapolations based on these parameters, and
  \cite{Herbreteau:IandC:2016} for simulations based on these
  parameters); $LU$-simulations are superior to $M$-simulations.

\item the way the parameters are computed: global parameters which
  associate a bound to each clock versus the more local state based
  parameters as in Definition~\ref{def:local-guards} which associate a
  set of bounds functions to each state~\cite{Behrmann:TACAS:2003};
  local bounds are superior to global bounds.

\item when diagonal constraints are present, whether zones get split
  or not: each time a zone gets split, new enumerations start from
  each of the new nodes; clearly, a no-splitting-of-zones approach is
  superior to zone splitting.
\end{itemize}

Algorithm of column $1$ uses the superior heuristic in all the
three optimizations above. The no-splitting-of-zones was possible
thanks to our simulation approach, which temporarily splits
zones for checking $Z \luglu Z'$, but never starts a new exploration from any
of the split nodes. The algorithm of column $2$, which is implemented
in the current version UPPAAL 4.1 uses the inferior heuristic in all
the three above. In particular, it is not clear how the extrapolation
approach can avoid the zone splitting in an efficient manner. The
superiority of our approach gets amplified (by multiplicative factors)
when we consider bigger products with many more diagonals.
In the third algorithm, we give a
diagonal free equivalent of the original model (c.f. Theorem~\ref{thm:diagonal-to-diagonalfree}) and use the UPPAAL
engine for diagonal free timed automata.
The UPPAAL diagonal free engine is highly
optimized, and makes use of the superior heuristics in the first two
optimizations mentioned above (the third is not applicable now as it
is a diagonal free automaton). The third algorithm can be considered as a
good approximation of the zone splitting approach to diagonal
constraints using $LU$-abstractions and local guards.

The second and
the third
methods are the only possibilities of verifying timed models coming
with diagonal constraints in UPPAAL. Both these approaches are in principle prone
to a $2^{\# \mathcal{D}}$ blowup compared to the first approach, where
$\mathcal{D}$ gives the number of diagonal constraints. The table
shows that a good extent of this blowup indeed happens.
The UPPAAL diagonal free engine uses ``minimal constraint
systems''~\cite{Bengtsson:Springer:2004} for representing zones,
whereas in our implementation we use Difference Bound Matrices
(DBMs)~\cite{Dill:1990:DBM}. This explains why even with more nodes
visited, UPPAAL performs better in terms of time in some cases.  We
have not included in the table the comparison with two other works
dealing with the same problem: the refined diagonal free conversion
\cite{Reynier:toolreport} and the $LU$ simulation extended with $LU^d$
simulation for diagonals \cite{Gastin:2018:CONCUR}. However, our
results are better than the tables reported in these papers.

\section{Conclusion}
\label{sec:conclusion}

We have proposed a new algorithm for handling diagonal constraints in
timed automata, and extended it to automata with general updates. Our
approach is based on a simulation relation between zones.  From our
preliminary experiments, we can infer that the use of simulations is
indispensable in the presence of diagonal constraints as zone-splitting
can be avoided. Moreover, the fact that the simulation approach stores
the actual zones (as opposed to abstracted zones in the extrapolation
approach) has enabled optimizations for diagonal-free automata that work with dynamically
changing simulation parameters ($LU$-bounds), which are learnt as and
when the the zones are expanded \cite{Herbreteau:2013:CAV}. Working
with actual zones is also convenient for finding cost-optimal paths in
priced timed automata~\cite{Bouyer::2016:CAV}. Investigating these in
the presence of diagonal constraints is part of future work. Currently, we have not
implemented our approach for updateable timed automata. This will also
be part of our future work.

Working directly with a model containing
diagonal constraints could be convenient (both during modeling, and
during extraction of diagnostic traces) and can also potentially give
a smaller automaton to begin with. We believe that our experiments
provide hope that diagonal constraints can indeed be used.


\bibliographystyle{plainurl}
\bibliography{ms}

\end{document}